\documentclass[11pt]{article}
\usepackage{times}
\usepackage[tbtags]{amsmath}
\usepackage{amsthm}
\usepackage{clrscode}
 \theoremstyle{plain}
 \newtheorem{theorem}{Theorem}
 \newtheorem{lemma}[theorem]{Lemma}
 \newtheorem{corollary}[theorem]{Corollary}
 \newtheorem{proposition}[theorem]{Proposition}
 
 \theoremstyle{definition}

 \theoremstyle{remark}

 \theoremstyle{plain}
 \newtheorem*{theorem*}{Theorem}
 \newtheorem*{lemma*}{Lemma}
 \newtheorem*{corollary*}{Corollary}
 \newtheorem*{proposition*}{Proposition}
 \newtheorem*{claim*}{Claim}
 
\usepackage{amssymb}
\usepackage{amsfonts}
\newlength{\actualtopmargin}
\newlength{\actualsidemargin}
\newcommand{\ket}[1]{| #1 \rangle}
\newcommand{\abs}[1]{\vert #1 \vert}
\newcommand{\ceil}[1]{\lceil #1 \rceil}

\newcommand{\bound}[2]{\in\{ #1,\ldots,#2\}}

\newcommand{\Vertice}{V}
\newcommand{\Edge}{E}
\newcommand{\Int}{\mathbb Z}
\newcommand{\myvec}[1]{\vec #1}

\newcommand{\calB}{\mathcal{B}}

\newcommand{\calH}{\mathcal{H}}

\newcommand{\sfQ}{\mathsf{Q}}
\newcommand{\sfR}{\mathsf{R}}
\newcommand{\sfS}{\mathsf{S}}
\newcommand{\sfT}{\mathsf{T}}

\newcommand{\bbC}{\mathbb{C}}

\newcommand{\bbF}{\mathbb{F}}

\newcommand{\bbQ}{\mathbb{Q}}

\newcommand{\Complex}{\bbC}

\newcommand{\bmzero}{\boldsymbol{0}}
\newcommand{\bmone}{\boldsymbol{1}}

\newcommand{\CNOT}{\mathrm{CNOT}}

\newcommand{\unitaryU}{\mathit{U}}
\newcommand{\unitaryW}{\mathit{W}}
\newcommand{\unitaryY}{\mathit{Y}}

\begin{document}

\sloppy

% ---------------------------------------------------------------------------
%   Title page
% ---------------------------------------------------------------------------

\title{\Large
 \textbf{
   General Scheme for Perfect Quantum Network Coding \\
   with Free Classical Communication
 }\\
}

\author{
 Hirotada Kobayashi\footnotemark[1]~~\footnotemark[2]\\
%%   \texttt{hirotada@nii.ac.jp}
 \and
 Fran\c{c}ois Le Gall\footnotemark[2]\\
%%   \texttt{legall@qci.jst.go.jp}
 \and
 Harumichi Nishimura\footnotemark[3]\\
%%   \texttt{hnishimura@mi.s.osakafu-u.ac.jp}
 \and
 Martin R\"otteler\footnotemark[4]\\
%%   \texttt{mroetteler@nec-labs.com}
}

\date{}

\maketitle
\thispagestyle{empty}
\pagestyle{plain}
\setcounter{page}{0}

\vspace{-5mm}

\renewcommand{\thefootnote}{\fnsymbol{footnote}}

\begin{center}
{\large
 \footnotemark[1]%
 Principles of Informatics Research Division\\
 National Institute of Informatics, Tokyo, Japan\\
%%   2-1-2 Hitotsubashi, Chiyoda-ku, Tokyo 101-8430, Japan\\
 [2.5mm]
 \footnotemark[2]%
 Quantum Computation and Information Project\\
 Solution Oriented Research for Science and Technology\\
 Japan Science and Technology Agency, Tokyo, Japan\\
%%   5-28-3 Hongo, Bunkyo-ku, Tokyo 113-0033, Japan\\
 [2.5mm]
 \footnotemark[3]%
 Department of Mathematics and Information Sciences\\
 Graduate School of Science\\
 Osaka Prefecture University, Sakai, Osaka, Japan\\
%%   1-1 Gakuen-cho, Naka-ku, Sakai, Osaka 599-8531, Japan\\
 [2.5mm]
 \footnotemark[4]%
 NEC Laboratories America, Inc., Princeton, NJ, USA
%%   4 Independence Way, Suite 200, Princeton, NJ 08540, USA
}\\
[5mm]
{\large 11 August 2009}\\
[8mm]
\end{center}

\begin{abstract}
  This paper considers the problem of efficiently transmitting quantum states
  through a network. It has been known for some time that without
  additional assumptions it is impossible to achieve this task {\em
    perfectly} in general --- indeed, it is impossible even for the simple
  butterfly network. As additional resource we allow free classical
  communication between any pair of network nodes. 
%% We show  that
It is shown that 
perfect quantum network coding is achievable in this model whenever classical network coding is possible 
over the same network when replacing all quantum capacities by classical capacities. 
%% More precisely, we demonstrate that perfect quantum network coding 
More precisely, it is proved that perfect quantum network coding 
using free classical communication is possible over a network with $k$ 
source-target pairs if there exists a classical linear (or even vector-linear) coding scheme over a finite ring. 
Our proof is constructive in that we give explicit quantum coding operations for each network node. 
%% We also give an upper bound on the number of classical communication required 
This paper also gives an upper bound on the number of classical communication required 
in terms of $k$, the maximal fan-in of any network node, and the size of the network.
\end{abstract}
\setcounter{footnote}{0}\vspace{-8mm}

\clearpage

%=====================================================================   
\section{Introduction}   
%===================================================================== 
\subsection{Background}
Network coding was introduced by Ahlswede, Cai, Li and Yeung
\cite{ACLY:2000} to send multiple messages efficiently through a
network. Usually, the network itself is given as a weighted, directed
acyclic graph with the weights denoting the capacities of the edges. A
typical example is the butterfly network in
Figure~\ref{fig:classicalbutterfly}. In this example the task is to send
one bit from $s_1$ to $t_1$ and another bit from $s_2$ to $t_2$, where
each edge is a channel of unit capacity.  It is obviously
impossible to send two bits simultaneously by {\em routing} since the
edge between $n_1$ and $n_2$ becomes a bottleneck.  However, using
{\em coding} at the nodes as shown in Figure~\ref{fig:classicalbutterfly},
it is feasible to send the two bits as desired. Two fundamental
observations are in order: First, 
%whereas making perfect copies of (unknown) physical information is impossible, 
copying classical information is possible. In the example of the butterfly network, this is used for
the operations performed at nodes $s_1$, $s_2$, and $n_2$. Second,
and perhaps more importantly, information can be
encoded. In the example, this is used at nodes $n_1$, $t_1$, and $t_2$
where the XOR operation (i.e., addition over the finite field
$\mathbb{F}_2$) is applied.  After the seminal result
\cite{ACLY:2000}, network coding has been widely studied from both
theoretical and experimental points of view and many
applications %such as wireless communication
have been found (we refer to Refs.~\cite{FS:2007,HL:2008} for a good source of information on this topic).
\begin{figure}[t]
\begin{center}
\setlength{\unitlength}{0.25mm}
\begin{picture}(220,230)
%\linethickness{0.5pt}
\thinlines
\put( 20, 220){\makebox(0,0){$x_1$}}
\put(200, 220){\makebox(0,0){$x_2$}}
\put( 20, 210){\vector(0,-1){20}}
\put(200, 210){\vector(0,-1){20}}
%\linethickness{1pt}
\thicklines
\put( 20, 180){\circle{20}}
\put( 20, 180){\makebox(0,0){$s_1$}}
\put( 20, 170){\vector(0,-1){120}}
\put( 10, 110){\makebox(0,0){$x_1$}}
\put( 28.94, 175.53){\vector(2,-1){72.12}}
\put( 65, 165){\makebox(0,0){$x_1$}}
\put(200, 180){\circle{20}}
\put(200, 180){\makebox(0,0){$s_2$}}
\put(200, 170){\vector(0,-1){120}}
\put(210, 110){\makebox(0,0){$x_2$}}
\put(191.06, 175.53){\vector(-2,-1){72.12}}
\put(151, 165){\makebox(0,0){$x_2$}}
\put(110, 135){\circle{20}}
\put(110, 135){\makebox(0,0){$n_1$}}
\put(110, 125){\vector(0,-1){30}}
\put( 85, 112){\makebox(0,0){${x_1 \oplus x_2}$}}
\put(110,  85){\circle{20}}
\put(110,  85){\makebox(0,0){$n_2$}}
\put(101.06,  80.53){\vector(-2,-1){72.12}}
\put( 53,  75){\makebox(0,0){${x_1 \oplus x_2}$}}
\put(118.94,  80.53){\vector(2,-1){72.12}}
\put(163,  75){\makebox(0,0){${x_1 \oplus x_2}$}}
\put( 20,  40){\circle{20}}
\put( 20,  40){\makebox(0,0){$t_2$}}
\put(200,  40){\circle{20}}
\put(200,  40){\makebox(0,0){$t_1$}}
%\linethickness{0.5pt}
\thinlines
\put( 20,  30){\vector(-0,-1){20}}
\put(  20,  0){\makebox(0,0){$x_2$}}
%\put( 27.01,  32.99){\vector( 1,-1){12.99}}
%\put( 40,  10){\makebox(0,0){$x_2$}}
%\put(192.99,  32.99){\vector(-1,-1){12.99}}
%\put(180,  10){\makebox(0,0){$x_1$}}
\put(200,  30){\vector( 0,-1){20}}
\put(200,  0){\makebox(0,0){$x_1$}}
\end{picture}
\caption{
  The butterfly network and a classical linear coding protocol.
  The node $s_1$ (resp.~$s_2$) has for input a bit $x_1$ (resp.~$x_2$).
  The task is to send $x_1$ to $t_1$ and $x_2$ to $t_2$. 
  The capacity of each edge is assumed to be one bit. 
 % Our convention here is that $s_1$ (resp.~$s_2$) receives $x_1$ (resp.~$x_2$)
  %through a virtual incoming edge,
  %and that $t_1$ (resp.~$t_2$) has also a virtual outgoing edge
  %through which it should output $x_1$ (resp.~ $x_2$,).
  \label{fig:classicalbutterfly}
}
\end{center}
\end{figure}
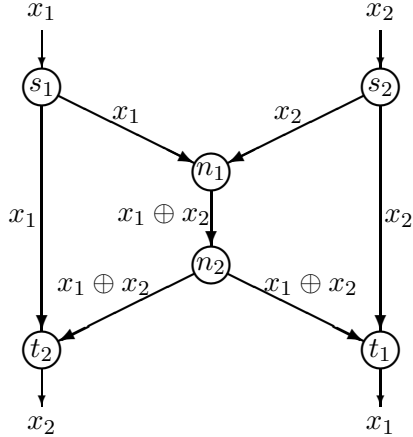

The multicast problem is a task that can be elegantly solved by network coding. 
In this problem, all messages at one source node must be sent to each of several target nodes. 
Ahlswede et al.~\cite{ACLY:2000} showed that the upper bound on the achievable rate given by the min-cut/max-flow condition 
is in fact always achievable. In other words, network coding allows one to send $m$ messages if and only if 
the value of any cut between the source node and each target node is at least $m$. 
Li, Yeung and Cai \cite{LYC:2003} showed that such a rate is always achievable 
by {\em linear} coding over a sufficiently large finite field (in which the operation performed at each node is a linear combination over some finite field).
%e.g., the operations over $\mathbb{F}_2$ performed at all nodes in Fig.~\ref{fig:classicalbutterfly}). 
Furthermore, this result was improved by Jaggi et al.\ \cite{JSC+:2005} who showed that such encoding can be constructed in polynomial time with respect 
to the number of nodes. 
This implies that deciding whether a given multicast network has a (linear) network coding scheme can be solved 
in polynomial time.  This contrasts with the general network coding problem for which 
Lehman and Lehman \cite{LL:2004} showed that it is NP-hard to decide whether there exists a linear coding solution.

Another important subclass of network coding problems is the {\em $k$-pair problem} 
(also called the {\em multiple-unicast problem}). 
In this setting the network has $k$ pairs of source/target nodes $(s_i,t_i)$, 
and each source $s_i$ wants to send a message $x_i$ to the target $t_i$. Notice that Figure~\ref{fig:classicalbutterfly} can be considered 
as a solution of a two-pair problem. It was shown by Dougherty and Zeger \cite{DZ:2006} 
that the solvability (resp.\ linear solvability) of any network coding 
problem can be reduced to the solvability (resp.\ linear solvability) of some instance of the $k$-pair problem.      
Combined with the Lehman-Lehman result, this implies that the linear solvability of the $k$-pair problem is NP-hard. 
Polynomial-time constructions of linear coding for fixed $k$ have been investigated \cite{INPRY:2008,WS:2007,WS:2007b}, 
but no complete answer has been obtained yet. 
 
Recently, network coding has become a topic of research in quantum
computation and information, giving rise to a theory of quantum
network coding.  The most basic setting is the following. The messages
are quantum states, and the network is also quantum, i.e., each edge
corresponds to a quantum channel.  The question is whether quantum
messages can be sent efficiently to the target nodes through the
network (possibly using the idea of network coding).  A very basic
difficulty immediately arises as opposed to the classical case:
quantum information cannot be copied \cite{NC:2000}. Hence,
multicasting quantum messages is impossible without imposing any extra
conditions. One approach to work around this problem has been
developed by Shi and Soljanin \cite{SS:2006}, who constructed a
perfect multicasting scheme over families of quantum networks under
the condition that the source owns many copies of quantum states.  A
more natural target may be the $k$-pair problem since here the number of inputs
matches the number of outputs. For this problem, however, there are
already a number of negative results known for the above basic
setting. First, Hayashi~et~al.~\cite{HIN+:2007} showed that sending
two qubits simultaneously and {\em perfectly (i.e., with fidelity
  one)} on the butterfly network is impossible. Leung, Oppenheim and
Winter \cite{LOW:2006} extended this impossibility result to the case
where the messages have to be sent in an asymptotically perfect way,
and also to classes of networks other than the butterfly network.
This means that some extra condition is needed to achieve perfect
quantum network coding for the $k$-pair problem case as well.

%\textbf{Main results.} 
\subsection{Main results} 
The extra condition considered in this paper is to allow \emph{free classical communication} 
to assist with sending quantum messages perfectly through the network. That is, any two nodes can 
communicate with each other through a classical channel which can be used freely (i.e., at no cost). Free 
classical communication as an extra resource often appears in quantum information theory, e.g.,~entanglement 
distillation and dilution (see Ref.~\cite{NC:2000}). Also, from a practical viewpoint, 
quantum communication is a very limited resource while classical communication is much easier 
to implement. Thus it would be desirable if the amount of quantum communication could be reduced 
using network coding with the help of classical communication. Another 
extra resource that may be considered (and rather popular in quantum information processing) is \emph{entanglement}, 
such as shared EPR pairs. However, it has the weakness that, once used, quantum communication is needed to recreate it. 
Therefore, allowing free classical communication arguably is a comparatively mild additional resource for perfect quantum network coding.

The first result of this paper (Theorem~\ref{theorem1}) can be summarized as follows: if there exists a classical linear coding scheme
over a ring $R$ for a $k$-pair problem given by a graph $G=(V,E)$, then there exists a solution
to the quantum $k$-pair problem over the same graph $G$ if free classical communication is allowed.
The idea to obtain this result is to perform a node-by-node simulation of the coding scheme 
solving the classical problem. For example, suppose that, in the classical coding scheme, a node $v$ of $G$ performs the map
$(z_1,z_2)\mapsto f(z_1,z_2)$ where $f \colon R^2\to R$ is some function. In the quantum case, this node
will perform the  quantum map $\ket{z_1,z_2}\ket{0}\mapsto \ket{z_1,z_2}\ket{f(z_1,z_2)}$.
A basic observation is that the first two registers should be ideally ``removed'' in order to simulate properly the classical
scheme. This task does not seem straightforward since the quantum state is in general a superposition of basis states, and this superposition has to be preserved
so that the input state can be recovered at target nodes.
Our key technique shows that this can be done if free classical communication is allowed, and if 
the classical scheme  to be simulated is linear.
More precisely, these registers are ``removed'' by measuring them in the Fourier basis associated with the additive group of~$R$.
An extra phase then appears, but it can be corrected \emph{locally} at each target node as will be proved in Section~\ref{section:main}.
This requires (free) classical communication. In our construction it is sufficient to send  at most $kM\abs{V}$ elements of $R$, 
where $M$ is the maximal fan-in of nodes of $G$. 
%and this superposition has to be preserved so that the input state can be recovered 
%by local operations at target nodes.
%Our key-technique (a basis change by the quantum Fourier transform over 
%the additive group of $R$)  
%shows that this can be done if free classical communication is allowed, and if 
%the classical scheme  to be simulated is linear.

A  classical coding scheme is called {\em vector-linear}
if the operation at each node is of the form $\sum_i A_i \myvec{v_i}$,
where the $\myvec{v_i}$'s are the vectors input to the node and $A_i$ is a matrix
to apply to $\myvec{v_i}$.
The result above can be extended to the vector-linear coding case as well (Theorem \ref{theorem2}).
That is, 
if there exists a classical vector-linear coding scheme over a ring $R$ for a $k$-pair problem given by a graph
$G$, then there also exists a solution to the quantum $k$-pair problem over the same graph $G$ (again if free classical communication is allowed).
Notice that there are examples of graphs over which a vector-linear solution is known but no
linear coding scheme exists (see Refs.~\cite{MEHK:2003,LL:2004}).
There are also examples
for which even vector-linear coding is not sufficient \cite{DFZ:2005}.
However, most of known networks solvable by network coding
have vector-linear solutions, and hence our result is applicable quite widely (and actually is even applicable to the examples in Ref.~\cite{DFZ:2005}).

%=== Removed paragraph 23 April 2009
%=== Restored 5 August 2009
Recently, the present authors studied a slightly generalized version of the $k$-pair problem 
where the pairs $(s_i,t_i)$ can be chosen at the end of the protocol \cite{KLNR:2009}. 
The strategy there worked only when there is a solution to the associated classical
multicast problem, in which each source node $s_i$ has to send its input to all target nodes $t_j$.
The idea was to simulate a classical multicast coding scheme over a finite field  
in order to create shared cat states (and then to convert them into EPR pairs 
so that quantum teleportations can be performed). 
This paper gives a more direct way of realizing quantum state transmission 
(without teleportations) that is applicable whenever there is a linear solution 
to the classical $k$-pair problem, which is a much weaker condition.
Furthermore, we extend a core idea of the protocol of Ref.~\cite{KLNR:2009}, 
namely, the quantum simulation of a classical linear coding scheme over a finite field, 
to a finite ring and to the vector-linear case.

%\textbf{Related work.}
\subsection{Related work}
There are several papers studying quantum network coding 
on the $k$-pair problem in situations different from the most basic setting 
(perfect transmission of quantum states using only a quantum network of limited capacity). 
Hayashi et al.\ \cite{HIN+:2007} and Iwama et al.\ \cite{INRY:2006} 
considered ``approximate'' transmission of qubits in the $k$-pair problem, 
and showed that transmission with fidelity larger than $1/2$ is possible for a class 
of networks. Hayashi~\cite{H:2007} showed how to achieve perfect transmission of two qubits 
on the butterfly network if two source nodes have prior entanglement,
and if, at each edge, we can choose between sending two classical bits and sending one qubit. 
Leung, Oppenheim and Winter~\cite{LOW:2006} considered various extra resources 
such as free forward/backward/two-way classical communication and entanglement, 
and investigated the lower/upper bounds of the rate of quantum network coding 
for their settings. The setting of the present paper is close to their model allowing free two-way classical communication. 
The difference is that Ref.~\cite{LOW:2006} considered asymptotically 
perfect transmission while this paper focuses on perfect transmission. 
Also, Ref.~\cite{LOW:2006} showed optimal rates for a few classes of networks 
while the present paper gives lower bounds for much wider classes of networks.
% (which may not always be optimal).
%Recently, the present authors \cite{KLNR:2009} gave a perfect quantum teleportation protocol in a setting that is close to
%the multicast model. The strategy was to first share cat states, and then to create EPR pairs so that we can perform quantum teleportations.
%This approach cannot work for the $k$-pair problem because in general such cat states cannot be created in this setting.
%This paper gives a more direct way of realizing quantum state transmission without teleportation that is applicable even to the $k$-pair problem.
%We also extend a core idea of the protocol of Ref.~\cite{KLNR:2009}, the quantum simulation of the linear coding over a finite field, 
%to a finite ring and the vector-linear case.

As mentioned in Ref.~\cite{LOW:2006}, free classical communication essentially makes 
the underlying directed graph of the quantum network undirected since quantum teleportation enables one
to send a quantum message to the reverse direction of a directed edge. In this context, our result 
gives a lower bound of the rate of quantum network coding that might not be optimal 
even if its corresponding classical coding is optimal in the directed graph.  
However, even in the classical case, network coding over an undirected graph 
is much less known than that over a  directed one. In the multicast network, the gap between 
the rates by network coding and by routing is known to be at most two~\cite{LLL:2008}, while there is an example for which
the min-cut rate bound cannot be achieved by network coding \cite{LLL:2008,ABY:2008}.
Also notice that, in the $k$-pair problem, it is conjectured that fractional routing achieves the optimal rate 
for any undirected graph (see for example Ref.~\cite{HKL:2006}).
However, this conjecture has been proved only for very few families of networks, and remains one of the main open problems in
the field of network coding.

%=====================================================================
\section{\boldmath{The $k$-pair problem}}\label{section:statement}
%=====================================================================
\subsubsection*{\boldmath{The classical $k$-pair problem.}}
%\noindent
We recall the statement of the $k$-pair problem in the classical case, and the definition of a solution to this problem.
The reader is referred to, for example, Ref.~\cite{DZ:2006} for further details. 

An instance of a $k$-pair problem is a directed graph ${G=(\Vertice,\Edge)}$ and $k$ pairs 
of nodes $(s_1,t_1),\ldots,(s_k,t_k)$. %The goal is to send $k$ messages 
%$x_1,\ldots,x_k$ from the source nodes to the target nodes.
For $i\bound{1}{k}$, $x_i$ is given at the source~$s_i$, and has to  
be sent to the target~$t_i$ through $G$ under the condition 
that each edge has unit capacity.
Let $\Sigma$ be a finite set.
%can send one element of $\Sigma$ (i.e., each edge has 
%unit capacity).
% where $\Sigma$ is an appropriately chosen finite set. 
A coding scheme over $\Sigma$ is a choice of operations for all nodes in $\Vertice$: 
for each node $v\in\Vertice$ with fan-in~$m$ and fan-out~$n$, the operation at $v$ 
is written as $n$ functions $f_{v,1},\ldots,f_{v,n}$, each from $\Sigma^m$ 
to $\Sigma$, where the value $f_{v,i}(z_1,\ldots,z_m)$ represents the message 
sent through the $i$-th outgoing edge of $v$ when the inputs of the $m$ incoming edges are
$z_1,\ldots,z_m$.  
A {\em solution} over $\Sigma$ to an instance of the $k$-pair problem is a coding scheme over $\Sigma$ 
that enables one to send simultaneously $k$ messages $x_i$ from $s_i$ to $t_i$, for all~$i\bound{1}{k}$.
For  example, the coding scheme in Figure~\ref{fig:classicalbutterfly} is a solution over $\mathbb{F}_2$ to the two-pair problem associated with the butterfly graph.

%The key result of this paper is
%a quantum simulation of any classical linear network coding scheme over a ring.
%% Here we use the standard definition of classical linear network coding (see~\cite{LYC:2003,SET:2003}).
For convenience, the following simple but very useful convention is assumed when describing a classical coding scheme. 
Each source ${s_i}$ is supposed to have a ``virtual'' incoming edge
from which it receives its input $x_i$.
Also, each target ${t_i}$ is supposed to have a ``virtual'' outgoing edge,
where $x_i$ must be output through.
In this way, the source and target nodes perform 
coding operations on their inputs,
and this convention enables one to ignore the distinction between
source/target nodes and internal nodes.
These conventions are illustrated in Figure~\ref{fig:classicalbutterfly}.

%===============================================
\subsubsection*{\boldmath{The quantum $k$-pair problem.}}
%===============================================
%\noindent 
We suppose that the reader is familiar with the basics of quantum information theory and 
refer to Ref.~\cite{NC:2000} for a good reference.
In this paper we will consider $d$-dimensional quantum systems, i.e.,~quantum 
states with associated complex Hilbert space
%~${\Complex^d}$
of dimension~$d$, for some positive integer $d$. 
%The simplest case of ${\calH = \Complex^2}$ is of particular importance, and a system supporting such a state space is called a \emph{qubit (quantum bit)}.
%For a $d$-dimensional system, we label the orthonormal basis states
%by the elements of some alphabet of size $d$, e.\,g., the numbers $\{0,1,\ldots,d-1\}$.
%A normalized vector in $\Complex^d$ is called a \emph{qudit},
%is written as ${\ket{\psi} = \sum_{i=0}^{d-1} \alpha_i \ket{i}}$,
%where ${\alpha_i \in \Complex}$ and ${\sum_{i=0}^{d-1} \abs{\alpha_i}^2 = 1}$.
%Multi-registers quantum states consist of several qudits.
%The basis states of a quantum state of $n$ qudits are tensor products of the basis states of the single qudits.
%A general state of a quantum register of $n$ qudits
%is a normalized vector in ${\calH = (\Complex^d)^{\otimes n} \cong \Complex^{d^n}}$.

An instance of a quantum $k$-pair problem is, as in the classical case, a directed graph $G=(\Vertice,\Edge)$ and $k$ pairs 
of nodes $(s_1,t_1),\ldots,(s_k,t_k)$. 
Let $\calH$ be a Hilbert space.
The goal is to send a quantum state $\ket{\psi}\in\calH^{\otimes k}$ supported on the source nodes $s_1,\ldots,s_k$ (in this order) to the target 
nodes $t_1,\ldots,t_k$ (in this order). We consider the model where each edge of  $G$ can transmit one quantum state over~$\calH$. 
However, free classical communication is allowed between any two nodes of~$G$. 
For a positive integer~$d$, an instance of the quantum $k$-pair
 problem is said \emph{solvable} 
%Let $d$ be a positive integer.
%We say that an instance of the quantum $k$-pair problem is solvable 
over $\Complex^d$ if there exists a protocol solving this problem for~$\calH=\Complex^d$.

%========================================================================
\section{Perfect quantum network coding}\label{section:main}
%========================================================================

%==========================
\subsection{Linear coding over rings}\label{subsection:ring}
%==========================

This subsection considers instances of the $k$-pair problem for which there exists a solution 
using classical linear coding over rings, i.e.,~$\Sigma$ is supposed to be a finite ring $R$ 
(not necessarily commutative). A coding scheme is said \emph{linear} over $R$ if the functions 
$f_{v,i}$ corresponding to the encoding operations performed at each node $v\in\Vertice$ are linear.
%Let us denote $\abs{R}=r$. 
%Recall that a classical linear coding scheme over $R$ consists of 
%linear codings $f_{v,i}$ at all nodes $v$ in a graph $G=(\Vertice,\Edge)$ given 
%as an instance of the $k$-pair problem.  

%A classical linear protocol over $R$ is a  protocol such that (1) at the end of the protocol, 
%each sink $t_i$ obtains $x_i$; and (2) each node performs a linear encoding as described in the following definition.
%\begin{definition}[{\bf Linear encoding at a node}]\label{definition1}
%Let $V$ be a node of $G$. Suppose that this node has  fan-in $m$ and fan-out $n$.
%Let $y_i$ for $i\bound{1}{m}$ denote the entries of the node.  
%A coding scheme for this node is a set of functions $f_i(y_1,\ldots,y_m)$ for $i\bound{1}{n}$, 
%each $f_i$ being associated to the $i$-th output of the node.  
%A linear coding scheme is a scheme such that the functions $f_i$ are linear functions of their inputs.
%\end{definition}

The main result of this subsection is the following theorem. %If $R$ is a ring, let $\abs{R}$ denote its size.
\begin{theorem}\label{theorem1}
Let $G=(\Vertice,\Edge)$ be a directed graph and $(s_1,t_1),\ldots,(s_k,t_k)$ be $k$ pairs of nodes.
Let $M$ be the maximal fan-in of nodes in $G$ and $R$ be a finite ring.
Suppose that there exists a linear solution over $R$ to the associated classical $k$-pair problem.
Then the corresponding quantum $k$-pair problem is solvable over $\Complex^{|R|}$. 
Moreover, there exists a quantum protocol for this task that sends at most $kM\abs{V}$ elements of $R$ as free classical communication, 
i.e., at most $kM\abs{V}\ceil{\log_2 \abs{R}}$ bits of classical communication.
\end{theorem}

The basic strategy for proving Theorem \ref{theorem1} is 
to perform a quantum simulation of the classical coding scheme.
%The first key idea of the simulation is to keep quantum superposition while the phase errors occur in the superposition.
%The second key idea is to correct the errors by local operations of target nodes with the help of free classical communication.
%This strategy 
This strategy is illustrated on the simple case $R=\bbF_2$ in Appendix A.
Before presenting the proof of this theorem, we need some preliminaries.

Let $\phi$ be a group isomorphism from the additive group of $R$ to some abelian group 
$A=\Int_{r_1}\times\cdots\times\Int_{r_\ell}$ with $\mathrm{\Pi}_{i=1}^\ell r_i=\abs{R}$ 
(but $\phi$ is not necessarily a ring isomorphism). There are many possibilities 
for the choice of $A$ and $\phi$. One convenient choice is to take $\Int_{r_1}\times\cdots\times\Int_{r_\ell}$ 
to be the invariant factor decomposition of the additive group of $R$. 
For any $x\in R$ and $i\bound{1}{\ell}$, let $\phi_i(x)$ denote the 
$i$-th coordinate of $\phi(x)$, i.e.,~an element of $\Int_{r_i}$. 
In the quantum setting, we suppose that each register contains a quantum state 
over $\calH=\Complex^{|R|}$, and denote by $\{\ket{z}\}_{z\in R}$ an orthonormal basis of $\calH$.
We define a unitary operator $\unitaryW$ over the Hilbert space $\calH$ as follows: 
for any $y\in R$, the operator $\unitaryW$ maps the basis state $\ket{y}$ to the state
%$
%\frac{1}{\sqrt {|R|}}\sum_{z\in R}\exp\bigl(2\pi \iota \sum_{i=1}^\ell\frac{\phi_i(y)\cdot\phi_i(z)}{r_i}\bigr)\ket{z}.
%$
$$
\frac{1}{\sqrt {|R|}}\sum_{z\in R}\exp\Bigl(2\pi \iota \sum_{i=1}^\ell\frac{\phi_i(y)\cdot\phi_i(z)}{r_i}\Bigr)\ket{z}.
$$
Here $\phi_i(y)\cdot\phi_i(z)$ denotes the product of $\phi_i(y)$ and $\phi_i(z)$, 
seen in the natural way as an element of the set $\{0,\ldots,r_i-1\}$. 
Note that $\unitaryW$ is basically the quantum Fourier transform over the additive group of $R$. 

Let $m$ and $n$ be two positive integers and $f_1,\ldots,f_n$ be $n$ functions from $R^m$ to $R$.
Let $\unitaryU_{f_1,\ldots,f_n}$ be the unitary operator over the Hilbert space ${\calH^{\otimes m}\otimes\calH^{\otimes n}}$ 
defined as follows:
%using the notations of Definition \ref{definition1} to denote the coding
%operation performed at node $V$, 
for any $m$ elements $y_1,\ldots,y_m$ and any $n$ elements $z_1,\ldots,z_n$ of $R$, 
the operator $\unitaryU_{f_1,\ldots,f_n}$ maps the basis state $\ket{y_1,\ldots,y_m}\ket{z_1,\ldots,z_n}$ to the state
$$\ket{y_1,\ldots,y_m}\ket{z_1+f_{1}(y_1,\ldots,y_m),\ldots,z_n+f_{n}(y_1,\ldots,y_m)}.$$ 
Now let us define the following quantum procedure $\proc{Encoding}(f_1,\ldots,f_n)$.

\begin{codebox}
\Procname{Procedure $\proc{Encoding}(f_1,\ldots,f_n)$}%\emph{[The notation of Definition \ref{definition1} are used]} }
\zi \const{input:} quantum registers $\sfQ_1,\ldots,\sfQ_m$,
% each containing an element of  $\calH$
each associated with $\calH$
\zi \const{output:} quantum registers $\sfQ'_1,\ldots,\sfQ'_n$, 
%each containing an element of  $\calH$, 
each each associated with $\calH$, and elements $a_1,\ldots,a_m$ of $R$ 
\li Introduce $n$ registers $\sfQ'_1,\ldots,\sfQ'_n$, each initialized to $\ket{0_{\calH}}$.
\li Apply the operator $\unitaryU_{f_1,\ldots,f_n}$ to 
%the registers 
$(\sfQ_1,\ldots,\sfQ_m,\sfQ'_1,\ldots,\sfQ'_n)$.
\li For each $i\bound{1}{m}$, apply $\unitaryW$ to 
%register 
$\sfQ_i$. 
\li Measure the first $m$ registers $\sfQ_1,\ldots,\sfQ_m$ in the computational basis.
\zi Let $a_1,\ldots,a_m\in R$ denote the outcomes of the measurements.
\li Output $\sfQ'_1,\ldots,\sfQ'_n$ and the $m$ elements $a_1,\ldots,a_m$.
\End
\end{codebox}\vspace{0mm}

The behavior of this procedure on a basis state is described in the next proposition.
\begin{proposition}\label{prop:nodesimulation}
%Suppose that $z_i$ is a linear function of $x_1,\ldots,x_k$.
Suppose that the contents of the registers $\sfQ_1,\ldots,\sfQ_m$ 
forms the state $\ket{y_1,\ldots,y_m}_{(\sfQ_1,\ldots,\sfQ_m)}$ for some elements $y_1,\ldots,y_m$ of $R$.
Then the state in $(\sfQ'_1,\ldots,\sfQ'_n)$ after applying Procedure $\proc{Encoding}(f_1,\ldots,f_n)$ 
is of the form
\[
\exp\left(2\pi \iota g(y_1,\ldots,y_m)\right)
\ket{f_{1}(y_1,\ldots,y_m),\ldots,f_{n}(y_1,\ldots,y_m)}_{(\sfQ'_1,\ldots,\sfQ'_n)},
\]
where $g:R^m\to \bbQ$ is an additive group homomorphism determined 
by the measurement outcomes $a_1,\ldots,a_m$.
\end{proposition}
\begin{proof}
%Since each input $z_i$ of $V$ is a linear function of the $x_j$, i.e.~for each $i\bound{1}{s}$, we can write
%$$z_i=\sum_{j=1}^k\alpha_{i,j}x_j,$$ where the $\alpha_{i,j}$'s are elements of $R$ (possibly zero).  
%After Step 1 the state is $\ket{z_1,\ldots,z_s}_{\sfR_1,\ldots,\sfR_s}\ket{f_1(z_1,\ldots,z_s),\ldots, f_t(z_1,\cdots,z_s)}_{\sfR'_1,\ldots,\sfR'_t},$.
After Step~3, the resulting state is 
\[
\begin{split}
\frac{1}{\sqrt{|R|^m}}
&
\sum_{z_1,\ldots,z_m\in R}
\exp\Bigl(2\pi \iota\sum_{i=1}^\ell\sum_{j=1}^m\frac{ \phi_i(y_j)\cdot\phi_i(z_j)}{r_i}\Bigr)
\\
&
\times \ket{z_1,\ldots,z_m}_{(\sfQ_1,\ldots,\sfQ_m)}\ket{f_{1}(y_1,\ldots,y_m),\ldots,f_{n}(y_1,\ldots,y_m)}_{(\sfQ'_1,\ldots,\sfQ'_n)}.
\end{split}
\]

At Step~4, if the measurement outcomes are $a_1,\ldots,a_m$, 
where each $a_i$ is an element of $R$, then the state in
%registers
$(\sfQ'_1,\ldots,\sfQ'_n)$ becomes
$$
\exp\Bigl(2\pi \iota\sum_{i=1}^{\ell}\sum_{j=1}^m\frac{ \phi_i(a_j)\cdot\phi_i(y_j)}{r_i}\Bigr)\ket{f_{1}(y_1,\ldots,y_m),\ldots,f_{n}(y_1,\ldots,y_m)}_{(\sfQ'_1,\ldots,\sfQ'_n)}.$$
%Using the linearity of each $\phi^{(i)}$ 
This can be rewritten as 
$$
\exp{\left(2\pi \iota g(y_1,\ldots,y_m)\right)}\ket{f_{1}(y_1,\ldots,y_m),\ldots,f_{n}(y_1,\ldots,y_m)}_{(\sfQ'_1,\ldots,\sfQ'_n)},$$
where 
%$$g_V(z_1,\ldots,z_s)=\sum_{i=1}^{\ell}\sum_{j=1}^s\sum_{r=1}^k\frac{ u_j^{(i)}\phi^{(i)}(\alpha_{j,r}x_r)}{m_i}.$$
$g(y_1,\ldots,y_m)=\sum_{i=1}^{\ell}\sum_{j=1}^m\frac{ \phi_i(a_j)\cdot\phi_i(y_j)}{r_i}.$
Notice that $g$ is an additive group homomorphism determined by the values of $a_1,\ldots,a_m$. 
\end{proof}

Now we are ready to give the proof of Theorem \ref{theorem1}.

\begin{proof}[Proof of Theorem \ref{theorem1}]
Let $G=(\Vertice,\Edge)$ be a graph on which there exists a linear solution over $R$
to the classical $k$-pair problem associated with the pairs $(s_i,t_i)$. For each node $v\in\Vertice$ 
with fan-in $m$ and fan-out $n$, let $f_{v,1},\ldots,f_{v,n}$ be the coding operations performed at node $v$
in such a solution, where each $f_{v,i}$ is from $R^m$ to $R$.
Suppose that the input state of the quantum task is
$$
\ket{\psi_S}_{(\sfS_1,\ldots,\sfS_k)}=\sum_{x_1,\ldots,x_k\in R}\alpha_{x_1,\ldots,x_k}\ket{x_1}_{\sfS_1}\otimes\cdots\otimes\ket{x_k}_{\sfS_k},
$$
where the $\alpha_{x_1,\ldots,x_k}$'s are complex numbers such that $\sum_{x_1,\ldots,x_k\in R}\abs{\alpha_{x_1,\ldots,x_k}}^2=1$.
Here, for each $i\bound{1}{k}$, $\sfS_i$ is a register owned by the node $s_i$.
% and the $\alpha_{x_1,\ldots,x_k}$ are complex$\sum_{x_1,\ldots,x_k\in R}\abs{\alpha_{x_1,\ldots,x_k}}^2=1$.

The strategy is to simulate the solution to the associated classical task node by node.
We shall show that the classical coding operation performed at a node with fan-in $m$ 
can be simulated by sending $km$ elements of $R$ using free classical communication. 
The general bound $kM\abs{V}$ claimed in the statement of the theorem then follows.

More precisely, let $v\in\Vertice$ be a node of $G$ with fan-in $m$ and fan-out $n$. 
The coding performed at node $v$ is simulated as follows: 
The quantum procedure $\proc{Encoding}(f_{v,1},\ldots,f_{v,n})$ is used on the $m$ quantum registers 
input to $v$ through its $m$ incoming edges. The procedure outputs $n$ registers and $m$ elements $a_1,\ldots,a_m$ of $R$. 
Then all the elements $a_1,\ldots,a_m$ are sent to each target node (via free classical communication), 
and the $n$ registers are sent along the $n$ outgoing edges of $v$. 
%This simulation of the classical coding operation performed at node $V$ sends $sk$ elements of $A$ (that can be seen as element of $R$ through the map $\phi^{-1}$). 
%Denote by $\calB$ this procedure and by $\calB(\ket{\psi_S}_{\sfS_1,\ldots,\sfS_k})$ the final state. 
%Our goal is to show that 
%$$\calB(\ket{\psi_S}_{\sfS_1,\ldots,\sfS_k})=\sum_{x_1,\ldots,x_k\in R}\alpha_{x_1,\ldots,x_k}\ket{x_1}_{\sfS_1}\otimes\cdots\otimes\ket{x_k}_{\sfS_k},$$
%where, for $i\bound{1}{k}$, each registers $\sfT_i$ is owned by the target node $t_i$. 
Such a simulation is done for all the nodes in $\Vertice$. 
We refer to Appendix A for an example illustrating this strategy.

%and Proposition \ref{prop:nodesimulation} ensures 
%that this simulation is faithful, up an additional phase factor. 
In what follows, we denote by $\calB$ this strategy.
%% and by $\calB(\ket{\psi_S}_{\sfS_1,\ldots,\sfS_k})$  the final state after applying $\calB$. 
We first describe the behavior of $\calB$ when the input is a basis state.
\begin{lemma}\label{lemma:basis}
Let $x_1,\ldots,x_k$ be $k$ elements of $R$.
Then 
%% $$\calB(\ket{x_1}_{\sfS_1}\otimes\cdots\otimes\ket{x_k}_{\sfS_k})=e^{2\pi\iota h(x_1,\ldots,x_k)}\ket{x_1}_{\sfT_1}\otimes\cdots\otimes\ket{x_k}_{\sfT_k},$$
the state after applying $\calB$ to $\ket{x_1}_{\sfS_1}\otimes\cdots\otimes\ket{x_k}_{\sfS_k}$ is of the form
$$e^{2\pi\iota h(x_1,\ldots,x_k)}\ket{x_1}_{\sfT_1}\otimes\cdots\otimes\ket{x_k}_{\sfT_k},$$
where $h \colon R^k\to\bbQ$ is an additive group homomorphism depending only on the outcomes of the measurements done during the procedure.
Here, for each $i\bound{1}{k}$, the register $\sfT_i$ is owned by the target node $t_i$. 
\end{lemma}
\begin{proof}[Proof of Lemma \ref{lemma:basis}]
Since the classical coding scheme is linear, Proposition \ref{prop:nodesimulation} ensures that, 
at any step of the protocol, the quantum state of the system is of the form
\begin{equation}\label{eq:generalform}
\beta\ket{y_1}_{\sfQ_1}\otimes\cdots\otimes\ket{y_D}_{\sfQ_D}
\end{equation}
for some positive integer $D$ (depending on the step of the protocol) 
and some phase $\beta$ (depending on the step of the protocol, the outcomes of the measurements done, and the values $x_1,\ldots,x_k$),
such that %for any $i\bound{1}{d}$, 
each $y_i\in R$ can be written as a linear combination of the $x_j$'s, in a way that corresponds to the associated classical coding scheme. Here each $\sfQ_i$ is a register owned by some node of the graph $G$.

Let us get back to the simulation at node $v$ described above to work out the general form of the phase $\beta$.
Suppose that the current state of the quantum system is given by Eq.~(\ref{eq:generalform}). 
Suppose, without loss of generality, that the coding at node $v$ is done on the first $m$ registers. 
In other words, the simulation performed at node $v$ is done over the state 
$\beta\ket{y_1,\ldots,y_m}_{(\sfQ_1,\ldots,\sfQ_m)}\otimes \ket{y_{m+1},\ldots,y_D}_{(\sfQ_{m+1},\ldots,\sfQ_{D})}$ 
where each $y_i=\sum_{j=1}^k \gamma_{i,j}x_j$ for some constants $\gamma_{i,j}\in R$ (depending on the step of the protocol).
Then, from Proposition \ref{prop:nodesimulation}, the simulation done at this step (using the procedure $\proc{Encoding}(f_{v,1},\ldots,f_{v,n})$) can be seen
as transforming this state into the state 
\[
%\begin{split}
\beta e^{\left(2\pi \iota h_v(x_1,\ldots,x_k)\right)} 
%& 
\ket{f_{v,1}(y_1,\ldots,y_m),\ldots,f_{v,n}(y_1,\ldots,y_m)}_{(\sfQ'_1,\ldots,\sfQ'_n)}
%\\
\otimes 
%& 
\ket{y_{m+1},\ldots,y_D}_{(\sfQ_{m+1},\ldots,\sfQ_{D})},
%\end{split}
\]
where, for $g$ denoting the function in the statement of Proposition \ref{prop:nodesimulation},
$$h_v(x_1,\ldots,x_k)=g\Bigl(\sum_{j=1}^k\gamma_{1,j}x_j,\ldots,\sum_{j=1}^k\gamma_{m,j}x_j\Bigr).$$
Since $g$ is a group homomorphism, the function $h_v \colon R^k\to\bbQ$ is a group homomorphism also.

From the observation that the classical coding scheme solves the associated classical $k$-pair problem, 
we conclude that
%% $$
%% \calB(\ket{x_1}_{\sfS_1}\otimes\cdots\otimes\ket{x_k}_{\sfS_k})=e^{2\pi \iota \sum_{v\in \Vertice} h_v(x_1,\ldots,x_k)}\ket{x_1}_{\sfT_1}\otimes\cdots\otimes\ket{x_k}_{\sfT_k}.
%% $$
the state after applying $\calB$ can be written as
$$
e^{2\pi \iota \sum_{v\in \Vertice} h_v(x_1,\ldots,x_k)}\ket{x_1}_{\sfT_1}\otimes\cdots\otimes\ket{x_k}_{\sfT_k}.
$$
The claimed form is obtained by defining the function $h$ as $h(x_1,\ldots,x_k)=\sum_{v\in \Vertice} h_v(x_1,\ldots,x_k)$.
Notice that $h$ is determined only by the values of the measurements (the constants $\gamma_{i,j}$ are fixed by the choice of the classical coding scheme). 
\end{proof}
Now we proceed with the proof of Theorem \ref{theorem1}. Lemma \ref{lemma:basis} implies that 
%% $$\calB(\ket{\psi_S}_{\sfS_1,\ldots,\sfS_k})=
%% \sum_{x_1,\ldots,x_k\in R}\alpha_{x_1,\ldots,x_k}e^{2\pi \iota h(x_1,\ldots,x_k)}
%% \ket{x_1}_{\sfT_1}\otimes\cdots\otimes\ket{x_k}_{\sfT_k},
%% $$
the state after applying $\calB$ must be of the form
$$
\sum_{x_1,\ldots,x_k\in R}\alpha_{x_1,\ldots,x_k}e^{2\pi \iota h(x_1,\ldots,x_k)}
\ket{x_1}_{\sfT_1}\otimes\cdots\otimes\ket{x_k}_{\sfT_k},
$$
where, for each $i\bound{1}{k}$, the register $\sfT_i$ is owned by the target node~$t_i$. 
Also, Lemma \ref{lemma:basis} guarantees that each target node $t_i$ knows the function $h$ 
since the values of all the measurement have been sent to it.
Since $h$ is an additive group homomorphism, it can be written as $h(x_1,\ldots,x_k)=h_1(x_1)+\cdots+h_k(x_k)$, where
%% , for each $i\bound{1}{k}$,
the function $h_i \colon R\to\bbQ$ maps $x_i$ to $h(0,\ldots,0,x_i,0,\ldots,0)$, for each $i\bound{1}{k}$.

Now for each $i\bound{1}{k}$ the target node $t_i$ applies the map $\unitaryY_i$ to its register~$\sfT_i$, where $\unitaryY_i$ is defined as
\begin{eqnarray*}
\unitaryY_i\colon\ket{x}&\mapsto&
e^{-2\pi \iota h_i(x)}\ket{x},
\end{eqnarray*}
for any $x\in R$.
This step corrects the phases and the resulting state is $\ket{\psi_S}_{(\sfT_1,\ldots,\sfT_k)}$.
This concludes the proof of Theorem \ref{theorem1}.
\end{proof}

In Theorem \ref{theorem1}, for the clarity of the proof, we gave the bound ${kM|V|\lceil\log_2 |R|\rceil}$ 
of the number of classical bits to be sent. For concrete networks, this bound can be improved significantly: 
(i) at each node, the measurement outcomes $a_1,\ldots,a_m$ have only to be sent to the target nodes $t_j$ 
such that $\gamma_{i,j}\neq 0$ for some index $i\bound{1}{m}$, and 
(ii) any node performing only a copy operation 
does not require any free classical communication to be simulated quantumly. 
Furthermore, we can reduce the amount of classical communication to $O(1)$ 
in the subclass of $k$-pair problems considered in Ref.~\cite{INPRY:2008}, as shown in the following corollary.

\begin{corollary}
Suppose that there exists a classical linear coding scheme over a constant-size finite field $\mathbb{F}$ 
that solves a $k$-pair problem for a fixed constant $k$. 
Then the corresponding quantum $k$-pair problem is solvable over $\Complex^{\abs{\mathbb{F}}}$ 
by sending at most $O(1)$ elements 
of $\mathbb{F}$ as free classical communication.
\end{corollary}
\begin{proof}
Iwama et al.~\cite{INPRY:2008} showed that if $k$ and $|\mathbb{F}|$ are constant, 
we can find a classical linear coding scheme such that the total number of 
non-trivial linear operations is a constant (only depending on $k$ and $|\mathbb{F}|$).   
Then the corollary follows from Theorem \ref{theorem1}.
\end{proof}

%================================
\subsection{Vector-linear coding}
%================================
This subsection shows how to simulate classical vector-linear coding over any ring $R$ 
(possibly not commutative). This is one of the most general settings considered in the literature (see Ref.~\cite{DFZ:2005} for instance).  
%This is a generalization of the work of the previous section.

%This can also be seen as as performing coding over the $R$-module $R^n$.
%Let $G$ and $(s_i,t_i)$ be an instance of the classical $k$-pair problem. 
Let $R$ be a finite ring. Let $\Sigma$ be the $R$-module $R^q$ for some positive integer~$q$. 
Informally, this module is the analogue of the usual vector space $\mathbb{F}^q$ of dimension~$q$ 
over a finite field $\mathbb{F}$, but here $\mathbb{F}$ is replaced by the ring $R$. 
%As described in Section \ref{section:statement}, the goal of classical coding is 
%to send,  for $i\in\{1,\ldots,k\}$, a message $\myvec{x_i}\in \Sigma$ from the source $s_i$ to the targets $t_i$. 
%for each $i\in\{1,\ldots,k\}$ \footnote{In this subsection we use the notation $\myvec{x}$ to denote elements of the module $\Sigma=R^q$}.
A coding scheme is said \emph{$q$-vector-linear} over $R$ if, for each function
$f_{v,i}\colon (R^q)^m\to R^q$ corresponding to the $i$-th encoding operation performed at the node $v\in\Vertice$,
there exist matrices~$B^{(v,i)}_j$ of size $q\times q$ over $R$ such that $f_{v,i}(\myvec{y_1},\ldots,\myvec{y_m})=\sum_{j=1}^m B^{(v,i)}_j\myvec{y_j}$
for all $(\myvec{y_1},\ldots,\myvec{y_m})\in(R^q)^m$.

Again let $\phi$ be a group isomorphism from $R$ to an abelian group $\Int_{r_1}\times\cdots\times\Int_{r_\ell}$, and, 
for any $x\in R$ and $i\bound{1}{\ell}$, denote by $\phi_i(x)$ the $i$-th coordinate of~$\phi(x)$, i.e.,~an element of $\Int_{r_i}$.
Given an element $\myvec{x}=(x_1,\ldots,x_q)\in R^q$, let $\psi_i(\myvec{x})$ denote the element $(\phi_i(x_1),\ldots,\phi_i(x_q))$
in $\Int_{r_i}^q$ corresponding to the projections of all the coordinates of 
$\myvec{x}$ to $\Int_{r_i}$. 

The following theorem is proved in a manner similar to Theorem \ref{theorem1}.
\begin{theorem}\label{theorem2}
Let $G=(\Vertice,\Edge)$ be a directed graph and $(s_1,t_1),\ldots,(s_k,t_k)$ be $k$ pairs of nodes.
Let $M$ be the maximal fan-in of nodes in $G$, $R$ be a finite ring and~$q$ be a positive integer.
Suppose that there exists a $q$-vector-linear solution over $R$ to the associated classical $k$-pair problem.
Then the corresponding quantum $k$-pair problem is solvable over $\Complex^{|R|^q}$. 
Moreover, there exists a quantum protocol that sends at most $kM\abs{V}$ elements of $R^q$ as free classical communication.
\end{theorem}
\begin{proof}
%Note that each register contains a quantum state over a Hilbert space $\calH$ of dimension $\abs{R}^q$. 
All the results of Subsection \ref{subsection:ring} hold similarly by using the following Fourier transform $\unitaryW'$ 
instead of $\unitaryW$. The unitary operator $\unitaryW'$ is defined over the Hilbert space $\calH=\Complex^{|R|^q}$ by its action on the basis states of $\calH$: 
for any $\myvec{y}\in R^q$, the operator $\unitaryW'$ maps the state $\ket{\myvec{y}}$ to the state
$$\frac{1}{\sqrt {|R|^q}}\sum_{\myvec{z}\in R^q}\exp\Bigl(2\pi \iota \sum_{i=1}^\ell\frac{\psi_i(\myvec{y})\cdot\psi_i(\myvec{z})}{r_i}\Bigr)\ket{\myvec{z}},$$
where $\psi_i(\myvec{y})\cdot\psi_i(\myvec{z})$ denotes the inner product of the vectors $\psi_i(\myvec{y})$ and $\psi_i(\myvec{z})$.
If we denote by $A$ the abelian group of $R$, then $\unitaryW'$ is basically the quantum Fourier transform over the abelian group $A^q$.
\end{proof}

%\subsection{Applications}\begin{proposition}
%Suppose that there exists a classical vector linear coding scheme over a finite field $F$ of constant size.
%Then quantum coding can be performed on the same graph by sending at most $O(dk)$ elements of $F$ as free classical communication.
%\end{proposition}\begin{proof}
%From a result by Iwama et al.~\cite{INPRY:2008}, there exists a classical linear coding scheme requiring encoding at only a %constant numberof nodes. Then the result follows.\end{proof}

%=========================================================================================================
\section{Concluding remarks}
%=========================================================================================================
This paper has presented a protocol to achieve perfect quantum network coding with
free classical communication. The proposed protocol works for all $k$-pair
problems that can be solved by linear or by vector-linear coding over
any finite ring, encompassing a broad class of networks that 
have been studied classically. 

There are still several open problems. A natural question is whether
perfect quantum network coding (with free classical communication) is possible for any instance of the
$k$-pair problem solvable classically.  Another open problem is a
converse of the results of this paper, i.\,e., to determine whether
there exists an undirected network such that quantum coding is
possible (with free classical communication) but classical coding is
not possible. 
%Currently only an upper bound of the rate of quantum
%network coding is known in terms of the cut of the graph as given in
%Appendix B.

%\section*{Acknowledgments}
\subsubsection*{Acknowledgments}
\sloppy{
HK is partially supported by the Grant-in-Aid for Scientific Research~(B) Nos.~18300002~and~21300002
of the Ministry of Education, Culture, Sports, Science and Technology of Japan.
HN is partially supported by the Grant-in-Aid for Young Scientists (B)
No. 19700011 of the Ministry of Education, Culture, Sports, Science and
Technology of Japan.
}

%\bibliographystyle{IEEEtranS}
%\bibliographystyle{abbrv}
%\bibliography{KobLeGNisRotICALP09}

\begin{thebibliography}{10}

\bibitem{ACLY:2000}
R.~Ahlswede, N.~Cai, S.-Y.~R. Li, and R.~W. Yeung.
\newblock Network information flow.
\newblock {\em IEEE Transactions on Information Theory}, 46(4):1204--1216,
  2000.

\bibitem{ABY:2008}
A.~Al-Bashabsheh and A.~Yongacoglu.
\newblock On the capacity bound of undirected networks.
\newblock arXiv.org e-Print archive, arXiv:0804.4455, 2008.

\bibitem{DFZ:2005}
R.~Dougherty, C.~F. Freiling, and K.~Zeger.
\newblock Insufficiency of linear coding in network information flow.
\newblock {\em IEEE Transactions on Information Theory}, 51(8):2745--2759,
  2005.

\bibitem{DZ:2006}
R.~Dougherty and K.~Zeger.
\newblock Nonreversibility and equivalent constructions of multiple-unicast
  networks.
\newblock {\em IEEE Transactions on Information Theory}, 52(11):5067--5077,
  2006.

\bibitem{FS:2007}
C.~Fragouli and E.~Soljanin.
\newblock {\em Network Coding Fundamentals}.
\newblock Now Publishers Inc., 2007.

\bibitem{HKL:2006}
N.~Harvey, R.~Kleinberg, and A.~Lehman.
\newblock On the capacity of information networks.
\newblock {\em IEEE Transactions on Information Theory}, 52(6):2410--2424,
  2006.

\bibitem{H:2007}
M.~Hayashi.
\newblock Prior entanglement between senders enables perfect quantum network
  coding with modification.
\newblock {\em Physical Review A}, 76(4):040301(R), 2007.

\bibitem{HIN+:2007}
M.~Hayashi, K.~Iwama, H.~Nishimura, R.~Raymond, and S.~Yamashita.
\newblock Quantum network coding.
\newblock In {\em Proceedings of the 24th Annual Symposium on Theoretical
  Aspects of Computer Science}, volume 4393 of {\em Lecture Notes in Computer
  Science}, pages 610--621, 2007.

\bibitem{HL:2008}
T.~Ho and D.~Lun.
\newblock {\em Network Coding: An Introduction}.
\newblock Cambridge University Press, 2008.

\bibitem{INPRY:2008}
K.~Iwama, H.~Nishimura, M.~Paterson, R.~Raymond, and S.~Yamashita.
\newblock Polynomial-time construction of linear network coding.
\newblock In {\em Proceedings of the 35th International Colloquium on Automata,
  Languages and Programming}, volume 5125 of {\em Lecture Notes in Computer
  Science}, pages 271--282, 2008.

\bibitem{INRY:2006}
K.~Iwama, H.~Nishimura, R.~Raymond, and S.~Yamashita.
\newblock Quantum network coding for general graphs.
\newblock arXiv.org e-Print archive, quant-ph/0611039, 2006.

\bibitem{JSC+:2005}
S.~Jaggi, P.~Sanders, P.~A. Chou, M.~Effros, S.~Egner, K.~Jain, and
  L.~Tolhuizen.
\newblock Polynomial time algorithms for multicast network code construction.
\newblock {\em IEEE Transactions on Information Theory}, 51(6):1973--1982,
  2005.

\bibitem{KLNR:2009}
H.~Kobayashi, F.~{Le Gall}, H.~Nishimura, and M.~R\"otteler.
\newblock Perfect quantum network communication protocol based on classical
  network coding.
\newblock arXiv.org e-Print archive,
  arXiv:0902.1299, 2009.

%\bibitem{NCH}
%R.~Koetter.
%\newblock Network coding home page.
%\newblock http://tesla.csl.uiuc.edu/\~{}koetter/NWC/.

\bibitem{LL:2004}
A.~Lehman and E.~Lehman.
\newblock Complexity classification of network information flow problems.
\newblock In {\em Proceedings of the 15th ACM-SIAM Symposium on Discrete
  Algorithms}, pages 142--150, 2004.

\bibitem{LOW:2006}
D.~Leung, J.~Oppenheim, and A.~Winter.
\newblock Quantum network communication --- the butterfly and beyond.
\newblock arXiv.org e-Print archive, quant-ph/0608223, 2006.

\bibitem{LYC:2003}
S.-Y.~R. Li, R.~W. Yeung, and N.~Cai.
\newblock Linear network coding.
\newblock {\em IEEE Transactions on Information Theory}, 49(2):371--381, 2003.

\bibitem{LLL:2008}
Z.~Li, B.~Li, and L.~C. Lau.
\newblock A constant bound on throughput improvement of multicast network
  coding in undirected networks.
\newblock {\em IEEE Transactions on Information Theory}, 55(3):1016--1026, 2009.

\bibitem{MEHK:2003}
M.~M{\'e}dard, M.~Effros, T.~Ho, and D.~Karger.
\newblock On coding for non-multicast networks.
\newblock In {\em Proceedings of the 41st Annual Allerton Conference on
  Communication, Control and Computing}, pages 21--29, 2003.

\bibitem{NC:2000}
M.~A. Nielsen and I.~L. Chuang.
\newblock {\em Quantum Computation and Quantum Information}.
\newblock Cambridge University Press, 2000.

\bibitem{SS:2006}
Y.~Shi and E.~Soljanin.
\newblock On multicast in quantum networks.
\newblock In {\em Proceedings of the 40th Annual Conference on Information
  Sciences and Systems}, pages 871--876, 2006.

\bibitem{WS:2007}
C.-C. Wang and N.~B. Shroff.
\newblock Beyond the butterfly -- a graph-theoretic characterization of the
  feasibility of network coding with two simple unicast sessions.
\newblock In {\em Proceedings of the IEEE International Symposium on
  Information Theory}, pages 121--125, 2007.

\bibitem{WS:2007b}
C.-C. Wang and N.~B. Shroff.
\newblock Intersession network coding for two simple multicast sessions.
\newblock In {\em Proceedings of the 45th Annual Allerton Conference on
  Communication, Control and Computing}, pages 682--689, 2007.

\end{thebibliography}
%\section*{Appendix}

%========================================
%\section*{Appendix: Example}
%========================================
%%%%%%%%%%%%%%%%%%%%%%%%%%%%%%%%%%%
%\noindent {\Large \bf Appendix}\vspace{3mm}\\
%%%%%%%%%%%%%%%%%%%%%%%%%%%%%%%%%%%
%\noindent {\large \bf A.1 Example}\vspace{3mm}

\appendix

\section{Example for our Protocol}

In this section, we illustrate the techniques developed in Section \ref{section:main} 
for the butterfly network shown in Figure~\ref{fig:butterfly}. Our task is to send 
a quantum state $\ket{\psi_S}$ from the source nodes $s_1$ and $s_2$ to the target nodes $t_1$ and $t_2$. 
The task can be achieved perfectly, i.\,e., with fidelity one,
using the protocol given in Theorem~\ref{theorem1}.
We give the explicit details for this example.
More precisely, we describe how the protocol simulates
the classical linear coding scheme over $R=\bbF_2$ presented in Figure~\ref{fig:classicalbutterfly}.
%The protocol applies the fan-out operations at nodes~$s_1$,~$s_2$,~and~$n_2$,
%while performs appropriate quantum coding operations at nodes~$n_1$,~$t_1$,~and~$t_2$.
  %To simplify the description, we defer the measurements of
  %the internal nodes until the very end. 
  %We also use the qubits in a
  %specific order which should be obvious from the protocol and we
  %indicate the controlled operations as a shorthand $(a,b)$
  %representing ${\rm CNOT}^{(a,b)}$ between qubits $a$ and $b$: 
Hereafter, all the registers are assumed to be single-qubit registers, 
i.e., $\calH=\Complex^2$. All the registers are also supposed to be initialized to $\ket{0}$.
We denote by $\{\ket{z}\}_{z\in\bbF_2}$ an orthonormal basis of $\calH$. 
According to our protocol, the measurement results at each node are sent 
to both $t_1$ and $t_2$, and the measured registers are disregarded.

\begin{figure}[t]
\begin{center}
\setlength{\unitlength}{0.25mm}
\begin{picture}(220,230)
%\linethickness{0.5pt}
\thinlines
\put( 20, 220){\makebox(0,0){$\sfS_1$}}
\put(200, 220){\makebox(0,0){$\sfS_2$}}
\put( 20, 210){\vector(0,-1){20}}
\put(200, 210){\vector(0,-1){20}}
%\put(  20, 200){\vector( 0,-1){12.99}}
%\put(  20, 210){\makebox(0,0){$\sfS_1$}}
%\put( 40, 200){\vector(-1,-1){12.99}}
%\put( 40, 210){\makebox(0,0){$\sfS'_1$}}
%\put(180, 200){\vector( 1,-1){12.99}}
%\put(180, 210){\makebox(0,0){$\sfS_2$}}
%\put(220, 200){\vector(-1,-1){12.99}}
%\put(220, 210){\makebox(0,0){$\sfS'_2$}}
%% \put( 20, 220){\makebox(0,0){$\sfS'_1$}}
%% \put(200, 220){\makebox(0,0){$\sfS'_2$}}
%% \put( 20, 210){\vector(0,-1){20}}
%% \put(200, 210){\vector(0,-1){20}}
%\linethickness{1pt}
\thicklines
\put( 20, 180){\circle{20}}
\put( 20, 180){\makebox(0,0){$s_1$}}
\put( 20, 170){\vector(0,-1){120}}
\put( 10, 115){\makebox(0,0){$\sfR_1$}}
\put( 28.94, 175.53){\vector(2,-1){72.12}}
\put( 65, 170){\makebox(0,0){$\sfR_2$}}
\put(200, 180){\circle{20}}
\put(200, 180){\makebox(0,0){$s_2$}}
\put(200, 170){\vector(0,-1){120}}
\put(210, 115){\makebox(0,0){$\sfR_3$}}
\put(191.06, 175.53){\vector(-2,-1){72.12}}
\put(155, 170){\makebox(0,0){$\sfR_4$}}
\put(110, 135){\circle{20}}
\put(110, 135){\makebox(0,0){$n_1$}}
\put(110, 125){\vector(0,-1){30}}
\put(100, 115){\makebox(0,0){$\sfR_5$}}
\put(110,  85){\circle{20}}
\put(110,  85){\makebox(0,0){$n_2$}}
\put(101.06,  80.53){\vector(-2,-1){72.12}}
\put( 65,  75){\makebox(0,0){$\sfR_6$}}
\put(118.94,  80.53){\vector(2,-1){72.12}}
\put(155,  75){\makebox(0,0){$\sfR_7$}}
\put( 20,  40){\circle{20}}
\put( 20,  40){\makebox(0,0){$t_2$}}
\put(200,  40){\circle{20}}
\put(200,  40){\makebox(0,0){$t_1$}}
%\linethickness{0.5pt}
\thinlines
\put( 20,  30){\vector(-0,-1){20}}
\put(  20,  0){\makebox(0,0){$\sfT_2$}}
%\put( 27.01,  32.99){\vector( 1,-1){12.99}}
%\put( 40,  10){\makebox(0,0){$x_2$}}
%\put(192.99,  32.99){\vector(-1,-1){12.99}}
%\put(180,  10){\makebox(0,0){$x_1$}}
\put(200,  30){\vector( 0,-1){20}}
\put(200,  0){\makebox(0,0){$\sfT_1$}}
\end{picture}
\caption{
  Perfect quantum network coding with free classical communication through the butterfly network. 
Each edge has quantum capacity one. The task is to send a given input quantum state~$\ket{\psi_S}$ in ${(\sfS_1, \sfS_2)}$
  to ${(\sfT_1, \sfT_2)}$ in this order of registers.
  Here, the quantum register~$\sfS_1$~(resp.~$\sfS_2$) is possessed by the source node~$s_1$~(resp.~$s_2$),
  while the quantum register~$\sfT_1$~(resp.~$\sfT_2$) is possessed by the target node~$t_1$~(resp.~$t_2$).
%%   There are two communication scenarios, namely arising from
%%   (i) identifying the pairs $(s_1,t_1)$ and $(s_2,t_2)$,
%%   or (ii) identifying the pairs $(s_1,t_2)$ and $(s_2,t_1)$.
  The protocol given in Theorem~\ref{theorem1} realizes
  perfect quantum transmission of $\ket{\psi_S}$.
  Each $\sfR_i$ indicates the quantum register to be sent
  along the corresponding edge in the protocol.
  %The quantum registers~$\sfS'_1$,~$\sfS'_2$,~$\sfT'_1$,~and~$\sfT'_2$
  %possessed by the source nodes~$s_1$~and~$s_2$
  %and the target nodes~$t_1$~and~$t_2$, respectively,
  %are used at the stage of sharing the cat states.
  Overall, a total of seven qubits of communication are necessary
  to transfer the state from the source to the target registers.
  \label{fig:butterfly}
}
\end{center}
\end{figure}
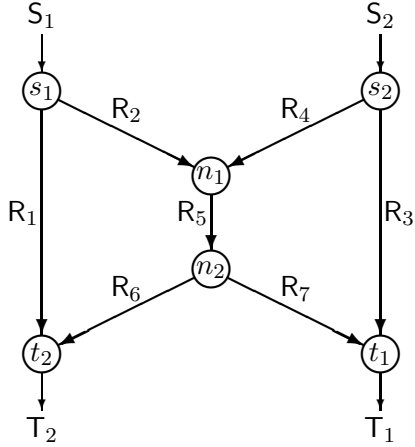

In our example, the quantum procedure $\proc{Encoding}(f_I,f_I)$ is applied at nodes $s_1,s_2$ and $n_2$, 
where $f_I$ denotes the identity map over $\bbF_2$. This procedure is implemented using
the following two unitary operators $\unitaryW$ and $\unitaryU_{f_I,f_I}$. The operator $\unitaryW$ is the Hadamard operator, 
and operator $\unitaryU_{f_I,f_I}$
maps each basis state $\ket{y}\ket{z_1,z_2}$ to the state $\ket{y}\ket{z_1+y,z_2+y}$.
%the measurement done in the Hadamard basis since it is equivalent 
%to the measurement in the computational basis $\{\ket{z}\}_{z\in\bbF_2}$ 
%in Step 4 after applying a Hadamard operator in Step 3. 
Procedure $\proc{Encoding}(f_+)$ is applied at nodes $n_1,t_1$ and $t_2$, 
where $f_+\colon(\bbF_2)^2\to\bbF_2$ is the function mapping $(y_1,y_2)$ to the element $y_1+y_2$ of $\bbF_2$.   
This procedure is implemented using $\unitaryW$ and the unitary operator $\unitaryU_{f_+}$ mapping 
$\ket{y_1,y_2}\ket{z}$ to $\ket{y_1,y_2}\ket{z+y_1+y_2}$. 
%for Step 2 and is also implemented by the controlled-NOT operator. 
%Moreover, Steps 3 and 4 corresponds to the measurement 
%in the Hadamard basis, which is implemented using $\calW$. 
Notice that both $\unitaryU_{f_I,f_I}$ and $\unitaryU_{f_+}$ can be implemented by using controlled-NOT operators
(hereafter, given two registers $\sfR$ and $\sfR'$, 
let $\CNOT^{(\sfR, \sfR')}$ denote the controlled-NOT operators mapping the basis state $\ket{z}_{\sfR}\ket{z'}_{\sfR'}$ 
to $\ket{z}_{\sfR}\ket{z+z'}_{\sfR'}$).

Now we present our protocol for the example of Figure~\ref{fig:butterfly}. Let 
$$
\ket{\psi_S}_{(\sfS_1,\sfS_2)}=
\alpha_{00}\ket{0}_{\sfS_1}\ket{0}_{\sfS_2}+\alpha_{01}\ket{0}_{\sfS_1}\ket{1}_{\sfS_2}+
\alpha_{10}\ket{1}_{\sfS_1}\ket{0}_{\sfS_2}+\alpha_{11}\ket{1}_{\sfS_1}\ket{1}_{\sfS_2}
$$
be the state that has to be sent from the source nodes to the target nodes. 
The two nodes $s_1$ and $s_2$ first implement the procedure $\proc{Encoding}(f_I,f_I)$ over their
registers $\sfS_1$ and $\sfS_2$, respectively. 
More precisely, at Step~1 of the procedure, $s_1$ (resp.~$s_2$) introduces 
two registers~$\sfR_1$~and~$\sfR_2$ (resp.~$\sfR_3$~and~$\sfR_4$), and, at Step~2, 
applies the operators $\CNOT^{(\sfS_1, \sfR_1)}$ and $\CNOT^{(\sfS_1, \sfR_2)}$
(resp.~$\CNOT^{(\sfS_2, \sfR_3)}$ and $\CNOT^{(\sfS_2, \sfR_4)}$) to implement ${\unitaryU}_{f_I,f_I}$. 
The resulting state is
\[
\begin{split}
\hspace{5mm}
&
%\hspace{-5mm}
 \alpha_{00}\ket{\bmzero}_{(\sfS_1, \sfR_1, \sfR_2)} \ket{\bmzero}_{(\sfS_2, \sfR_3, \sfR_4)}\\
  +&
  \alpha_{01}\ket{\bmzero}_{(\sfS_1, \sfR_1, \sfR_2)}\ket{\bmone}_{(\sfS_2, \sfR_3, \sfR_4)}\\
  +
  &
  \alpha_{10}\ket{\bmone}_{(\sfS_1, \sfR_1, \sfR_2)}\ket{\bmzero}_{(\sfS_2, \sfR_3, \sfR_4)}\\
  +&
  \alpha_{11}\ket{\bmone}_{(\sfS_1, \sfR_1, \sfR_2)}\ket{\bmone}_{(\sfS_2, \sfR_3, \sfR_4)}
.
\end{split}
\]
Hereafter, let $\bmzero$ and $\bmone$ denote strings of all-zero and all-one,
respectively, of appropriate length (three here). 
At Step~3, the operator $\unitaryW$ is applied to each register $\sfS_1$ and $\sfS_2$. At Step~4, these two registers are  
measured in the basis $\{\ket{z}\}_{z\in\bbF_2}$. Notice that the combination of Steps~3 ~and~4 corresponds to measurements in 
the Hadamard basis.
Let $a\in\bbF_2$ and $b\in\bbF_2$ denote the measurement outcomes.
%Now $\sfR_1$ and $\sfR_2$ (resp.~$\sfR_3$~and~$\sfR_4$) and $a$ (resp.~$b$) 
%are the output of $\proc{Encoding}(f_I,f_I)$ at node $s_1$ (resp.~$s_2$). 
The quantum state becomes
\[
\begin{split}
&
 \alpha_{00}\ket{\bmzero}_{(\sfR_1, \sfR_2)} \ket{\bmzero}_{(\sfR_3, \sfR_4)}\\
  +(-1)^{b}&\alpha_{01}\ket{\bmzero}_{(\sfR_1, \sfR_2)}\ket{\bmone}_{(\sfR_3, \sfR_4)}
  \\
  +(-1)^{a}&\alpha_{10}\ket{\bmone}_{(\sfR_1, \sfR_2)}\ket{\bmzero}_{(\sfR_3, \sfR_4)}
  \\
  +
  (-1)^{a+b}&\alpha_{11}\ket{\bmone}_{(\sfR_1, \sfR_2)}\ket{\bmone}_{(\sfR_3, \sfR_4)}
.
\end{split}
\]
Then the registers~$\sfR_1$~and~$\sfR_2$ are sent to $t_2$ and $n_1$, respectively,
while $\sfR_3$~and~$\sfR_4$ are sent to $t_1$ and $n_1$, respectively. 
The measurement outcomes $a$ and $b$ are sent to both target nodes by classical communication. 

Then the protocol proceeds with the simulation of the coding operation performed at node~$n_1$ in the classical coding scheme of
Figure~\ref{fig:classicalbutterfly} using the quantum procedure $\proc{Encoding}(f_+)$: Node~$n_1$ prepares a new register~$\sfR_5$ at Step~1, 
and applies the operators $\CNOT^{(\sfR_2, \sfR_5)}$ and $\CNOT^{(\sfR_4, \sfR_5)}$ to implement $\unitaryU_{+}$ at Step~2. 
The resulting state is
\[
\begin{split}
&
 \alpha_{00}\ket{\bmzero}_{(\sfR_1, \sfR_2)} \ket{\bmzero}_{(\sfR_3, \sfR_4)} \ket{0}_{\sfR_5}\\
  +
  (-1)^{b}&
  \alpha_{01}\ket{\bmzero}_{(\sfR_1, \sfR_2)} \ket{\bmone}_{(\sfR_3, \sfR_4)}\ket{1}_{\sfR_5}\\
  +
  (-1)^{a}&
  \alpha_{10}\ket{\bmone}_{(\sfR_1, \sfR_2)} \ket{\bmzero}_{(\sfR_3, \sfR_4)}\ket{1}_{\sfR_5}\\
  +
  (-1)^{a+b}&\alpha_{11}\ket{\bmone}_{(\sfR_1, \sfR_2)} \ket{\bmone}_{(\sfR_3, \sfR_4)}\ket{0}_{\sfR_5}
.
 \end{split}
\]
At Steps 3 and 4, the registers $\sfR_2$ and $\sfR_4$ are measured 
in the Hadamard basis. The measurement outcomes, denoted by $c_1$ and $c_2$, are sent 
to both target nodes. The quantum state becomes
 \[
 \begin{split}
&
  \alpha_{00}\ket{0}_{\sfR_1} \ket{0}_{\sfR_3} \ket{0}_{\sfR_5}\\
  +
  (-1)^{b+c_2}&
  \alpha_{01}\ket{0}_{\sfR_1}\ket{1}_{\sfR_3}\ket{1}_{\sfR_5}\\
  +
  (-1)^{a+c_1}&\alpha_{10}\ket{1}_{\sfR_1}\ket{0}_{\sfR_3}\ket{1}_{\sfR_5}\\
  +
  (-1)^{a+b+c_1+c_2}&\alpha_{11}\ket{1}_{\sfR_1}\ket{1}_{\sfR_3}\ket{0}_{\sfR_5}
  .
  \end{split}
\]
The register~$\sfR_5$ is then sent to $n_2$.

The node~$n_2$ now implements the procedure $\proc{Encoding}(f_I,f_I)$: It prepares 
two registers~$\sfR_6$~and~$\sfR_7$, applies the operators $\CNOT^{(\sfR_5, \sfR_6)}$ 
and $\CNOT^{(\sfR_5, \sfR_7)}$ to implement $\unitaryU_{f_I,f_I}$, and measures register $\sfR_5$ in the Hadamard basis. 
The outcome of this measurement, $d$, is sent to $t_1$ and $t_2$. The resulting state is
\[
\begin{split}
&
  \alpha_{00}\ket{0}_{\sfR_1} \ket{0}_{\sfR_3} \ket{\bmzero}_{(\sfR_6, \sfR_7)}\\
  +
  (-1)^{b+c_2+d}&
  \alpha_{01}\ket{0}_{\sfR_1} \ket{1}_{\sfR_3} \ket{\bmone} _{(\sfR_6, \sfR_7)}\\
  +
  (-1)^{a+c_1+d}&
  \alpha_{10}\ket{1}_{\sfR_1} \ket{0}_{\sfR_3} \ket{\bmone}_{(\sfR_6, \sfR_7)}\\
  + (-1)^{a+b+c_1+c_2}&
  \alpha_{11}\ket{1}_{\sfR_1} \ket{1}_{\sfR_3} \ket{\bmzero}_{(\sfR_6, \sfR_7)}
  ,
  \end{split}
\]
and the registers~$\sfR_6$~and~$\sfR_7$ are sent to $t_2$ and $t_1$, respectively.

In the last step of the simulation, both target nodes $t_1$ and $t_2$ apply the procedure $\proc{Encoding}(f_+)$: 
Node $t_1$~(resp.~$t_2$) prepares one register~$\sfT_1$ (resp.~$\sfT_2$),
and applies the operators $\CNOT^{(\sfR_3, \sfT_1)}$ and $\CNOT^{(\sfR_7, \sfT_1)}$ 
(resp.~$\CNOT^{(\sfR_1, \sfT_2)}$ and $\CNOT^{(\sfR_6, \sfT_2)}$) for $\unitaryU_{+}$. 
The resulting state is
\[
\begin{split}
&
\alpha_{00}\ket{0}_{\sfR_1} \ket{0}_{\sfR_3} \ket{\bmzero}_{(\sfR_6, \sfR_7)} \ket{0}_{\sfT_1} \ket{0}_{\sfT_2}\\
  +
  (-1)^{b+c_2+d}&\alpha_{01}\ket{0}_{\sfR_1}\ket{1}_{\sfR_3}\ket{\bmone}_{(\sfR_6, \sfR_7)}\ket{0}_{\sfT_1} \ket{1}_{\sfT_2}\\
  +
  (-1)^{a+c_1+d}&\alpha_{10}\ket{1}_{\sfR_1}\ket{0}_{\sfR_3}\ket{\bmone}_{(\sfR_6, \sfR_7)} \ket{1}_{\sfT_1} \ket{0}_{\sfT_2}\\
  +
   (-1)^{a+b+c_1+c_2}&\alpha_{11}\ket{1}_{\sfR_1}\ket{1}_{\sfR_3}\ket{\bmzero}_{(\sfR_6, \sfR_7)} \ket{1}_{\sfT_1} \ket{1}_{\sfT_2}\,.
\end{split}
\]
Then $t_1$ (resp.~$t_2$) measures registers $\sfR_3$ and $\sfR_7$ (resp.~$\sfR_1$ and $\sfR_6$) 
in the Hadamard basis. Let $e_1$ and $e_2$ (resp.~$f_1$ and $f_2$) be the outcomes of the measurements. 
The quantum state becomes
\[
\begin{split}
&
  \alpha_{00}\ket{0}_{\sfT_1} \ket{0}_{\sfT_2}\\
  +
  (-1)^{b+c_2+d+e_1+e_2+f_2}&\alpha_{01}\ket{0}_{\sfT_1} \ket{1}_{\sfT_2}\\
  +
  (-1)^{a+c_1+d+e_2+f_1+f_2}&\alpha_{10}\ket{1}_{\sfT_1} \ket{0}_{\sfT_2}\\
  +
   (-1)^{a+b+c_1+c_2+e_1+f_1}&\alpha_{11}\ket{1}_{\sfT_1} \ket{1}_{\sfT_2}.
\end{split}
\]

Now let  $h_1,h_2\colon\bbF_2\rightarrow\bbF_2$ be the functions defined by
\[
h_1(z)=(a+c_1+d+e_2+f_1+f_2)z\,,
\]
\[
h_2(z)=(b+c_2+d+e_1+e_2+f_2)z\,,
\]
for any $z\in\bbF_2$. The first key observation is that the above state can be rewritten as
\[
\sum_{z_1,z_2\in\bbF_2}(-1)^{h_1(z_1)+h_2(z_2)}\alpha_{z_1z_2}\ket{z_1}_{\sfT_1} \ket{z_2}_{\sfT_2}\,.
\]
The second key observation is that, since the nodes $t_1$ and $t_2$  received the outcomes of the intermediate measurements, 
they know the functions $h_1$ and $h_2$. 
Node $t_1$ (resp.~node $t_2$) then applies the quantum operation $\unitaryY_1$ (resp.~$\unitaryY_2$) 
mapping, for any $z\in\bbF_2$, 
the basis state $\ket{z}_{\sfT_1}$ to the state $(-1)^{h_1(z)}\ket{z}_{\sfT_1}$ (resp.~mapping $\ket{z}_{\sfT_2}$ to the state $(-1)^{h_2(z)}\ket{z}_{\sfT_2}$). 
The quantum state becomes 
$$\ket{\psi_S}_{(\sfT_1,\sfT_2)}=
\alpha_{00}\ket{0}_{\sfT_1} \ket{0}_{\sfT_2}
  +\alpha_{01}\ket{0}_{\sfT_1}  \ket{1}_{\sfT_2}
  +\alpha_{10}\ket{1}_{\sfT_1}  \ket{0}_{\sfT_2}
  +\alpha_{11}\ket{1}_{\sfT_1}  \ket{1}_{\sfT_2}.
$$
This completes our protocol.

\end{document}